\documentclass[11pt,letterpaper]{article}
\usepackage[margin=1in]{geometry} 
\usepackage{setspace}
\usepackage{orcidlink}
\usepackage{algorithm}
\usepackage{algpseudocode}
\usepackage{tikz}
\usepackage{graphicx}
\usepackage{hyperref}
\usepackage{caption}
\usepackage[labelfont=sf]{subcaption}
\usepackage{natbib}
\usepackage{amsmath}
\usepackage{amsthm}

\newtheorem{theorem}{Theorem}
\newtheorem{lemma}{Lemma}

\newtheorem{example}{Example}

\newcommand{\ubar}[1]{\text{\b{$#1$}}}

\begin{document}

\title{The Unreliable Job Selection and Sequencing Problem}

\author{
Alessandro Agnetis \orcidlink{0000-0001-5803-0438} \thanks{Dipartimento di
Ingegneria dell'Informazione e Scienze Matematiche, Universit\`a di Siena}
\and
Roel Leus \orcidlink{0000-0002-9215-3914}
\thanks{Research Centre for Operations Research and Statistics, KU Leuven}
\and
Emmeline Perneel \orcidlink{0000-0003-0504-024X}
\footnotemark[2]
\and 
Ilaria Salvadori \orcidlink{0009-0003-7248-8017}
\footnotemark[1]
}

\date{}
\maketitle
\medskip \noindent

\begin{abstract}
\noindent We study a stochastic single-machine scheduling problem, denoted the Unreliable Job Selection and Sequencing Problem (UJSSP). Given a set of jobs, a subset must be selected for processing on a single machine that is subject to failure. Each job incurs a cost if selected and yields a reward upon successful completion. A job is completed successfully only if the machine does not fail before or during its execution, with job-specific probabilities of success. The objective is to determine an optimal subset and sequence of jobs to maximize the expected net profit. We analyze the computational complexity of UJSSP and prove that it is NP-hard in the general case. The relationship of UJSSP with other submodular selection problems is discussed, showing that the special cases in which all jobs have the same cost or the same failure probability can be solved in polynomial time. To compute optimal solutions, we propose a compact mixed-integer linear programming formulation, a dynamic programming algorithm, and two novel stepwise exact algorithms. We demonstrate that our methods are capable of efficiently solving large instances by means of extensive computational experiments. We further show the broader applicability of our stepwise algorithms by solving instances derived from the Product Partition Problem.  
\end{abstract}

\vspace{2ex} \noindent{\textbf{Keywords: single-machine scheduling, job selection, unreliable jobs, submodular optimization, algorithm design}}

\section{Introduction}\label{sec:intro}
In many production and computing environments, machines are subject to random failures. In this paper, we consider a single-machine problem with permanent breakdowns. We are given a set $J = \{ 1, \ldots, n\}$ of jobs. Each job has a cost $c_j\geq 0$ which has to be paid if the job is selected, representing preparatory or setup expenses that occur regardless of whether the job completes successfully. A reward $r_j\geq 0$ is achieved if the job is successfully carried out, representing the profit, output value, or computational result obtained upon successful completion. The machine can fail during the execution of a job and when this happens, the job is lost and no further job can be carried out. The probability that the machine fails during the execution of job $j$ is given by $1-\pi_j$, where $\pi_j\in[0,1]$ is called the success probability of job $j$. If a certain subset $S$ of jobs is selected for processing, let $\sigma$ denote a sequence of these $|S|$ jobs, and let $\sigma(k)$ represent the $k$-th job in the sequence $\sigma$. The probability of reaching and successfully carrying out the $k$-th job in the sequence is given by 
\[
\prod_{i=1}^k \pi_{\sigma(i)},
\]
therefore the expected reward $R(S,\sigma)$ from selecting $S$ and sequencing its jobs according to $\sigma$ is
\begin{equation}\label{eq:exprofit}
R(S,\sigma)=\sum_{k=1}^{|S|}r_{\sigma(k)}\prod_{i=1}^k \pi_{\sigma(i)}.
\end{equation}
If we let $c(S)=\sum_{j\in S}c_j$ denote the total cost of selecting $S$, the \textit{expected net profit} of selecting $S$ and sequencing the jobs according to $\sigma$ is
\begin{equation}\label{eq:netprofit}
z(S,\sigma)=\sum_{k=1}^{|S|}r_{\sigma(k)}\prod_{i=1}^k \pi_{\sigma(i)}-c(S).
\end{equation}
We study the problem of selecting a subset $S\subseteq J$ of jobs and finding a sequence $\sigma$ so that the expected net profit $z(S,\sigma)$ is maximized. We call this problem the \emph{Unreliable Job Selection and Sequencing Problem (UJSSP)}.

When there are no costs ($c_j=0$ for all $j=1,\dots,n$), it is obviously profitable to select all jobs, and only the sequencing problem is left. In this case, we refer to the problem as the \emph{Unreliable Job Sequencing Problem (UJSP)}. This problem has been independently considered by various authors under different names and application contexts (see \cite{kadane1969}, \cite{stadje1995}, \cite{hellerstein1993}, and \cite{agnetis2009}) and is equivalent to the basic $n$-out-of-$n$ test sequencing problem, from which a considerable literature on testing problems originated (see \cite{unluyurt2025} for a recent overview). In this case, the optimal sequence is obtained by simply scheduling the jobs by non-increasing values of the following index \citep{kadane1969,mitten1960}:
\begin{equation}\label{eq:zrule}Z_j=\frac{\pi_jr_j}{1-\pi_j}.
\end{equation}
Consequently, in UJSSP, once a subset $S$ is selected, expression~\eqref{eq:netprofit} is maximized by sequencing the jobs in $S$ according to a schedule dictated by \eqref{eq:zrule}. Letting $\sigma_Z$ denote such sequence, in the following we let $R(S)=R(S,\sigma_Z)$ and $z(S)=z(S,\sigma_Z)$. We can restate UJSSP as the problem of finding $S^*$ such that:
\begin{equation*}
z(S^*)=\max_{S\subseteq J}\left\{R(S)-c(S)\right\}.
\end{equation*}
For the sake of simplicity, unless specified otherwise, we assume from now on that the $n$ jobs are numbered by non-increasing values of $Z_j$.

In this paper we present various results concerning UJSSP\@. In Section~\ref{sec:lit}, we discuss the most relevant research streams related to our model and briefly comment on practical contexts where such problems may arise. In Section~\ref{sec:compl}, we establish the complexity of UJSSP, showing that the problem is NP-hard in general through a reduction from the Product Partition Problem but polynomially solvable in some special cases. In Section~\ref{sec:milp} we present a compact mixed-integer linear programming (MILP) formulation of UJSSP\@. We introduce a dynamic programming (DP) algorithm in Section~\ref{sec:dp} and two novel stepwise exact algorithms in Section~\ref{sec:ex}. Detailed computational experiments for the algorithms are described in Section~\ref{sec:comp}, along with results and comments. Finally, some conclusions are drawn in Section~\ref{sec:conc}.

\section{Related work}\label{sec:lit}
Research on scheduling with machine unavailability has attracted increasing attention since the early 1990s (see \cite{schmidt2000} for an overview). Much of the literature considers settings where non-availability periods are known in advance and can be incorporated into the schedule. In many practical systems, however, machine failures are unpredictable, motivating the study of stochastic or online models in which breakdowns occur randomly. \cite{Glazebrook1984} proposed one of the first stochastic breakdown models on a single machine. He considered general cases in which the scheduling process is viewed as a Markov decision process with random job and machine breakdown durations. Policies based on Gittins’ index were used to decide which job to process at any time to maximize expected reward. More recent work on unexpected breakdowns on a single machine include \cite{huo2014}, who minimize total weighted completion time, and \cite{kacem2016}, who minimize either makespan or maximum lateness when jobs can have release and delivery times. 
In these models, breakdown intervals are often represented as \([s,t]\), where \(s\) and \(t\) are random variables following some probability distribution. In contrast, we assume that the probability of failure depends on the job being processed. This excludes failures caused by external preemptions but captures situations in which task characteristics (duration, complexity, resource intensity) influence reliability. Recent work in maintenance modeling supports this assumption: \cite{converso2023}, for instance, show that incorporating workload information improves failure-probability forecasts. Our model further assumes non-recoverable breakdowns: once the machine fails, it becomes unavailable for the remainder of the planning horizon. Permanent breakdown models are more appropriate for situations where the cost or time to repair is prohibitive and the resource cannot be restored within the decision horizon (for example for modelling catastrophic hardware failures). 

When machine breakdowns are assumed to be non-recoverable, the set of completed jobs becomes stochastic and may even be empty, making classical objectives such as makespan or total tardiness inappropriate. Instead, the goal is to maximize the expected revenue earned before failure. The problem of assigning \(n\) jobs with job-specific success probabilities and rewards to \(m\) machines under this objective was first studied by \cite{agnetis2009}, who showed its equivalence to the Total Weighted Discounted Completion Time Problem \citep{pinedo2002}, proved NP-hardness even for \(m=2\), and proposed a round-robin heuristic with a tight \(\frac{1}{2}\)-approximation ratio. Subsequent work on this problem examined list-scheduling algorithms for two and multiple machines \citep{agnetis2014, agnetis2020} while a variant in which the probability of failure of a job depends on its revenue and the machine it is processed on was examined by \cite{agnetis2017}. A similar problem in which a job has one replication for each machine and yields revenue if at least one execution succeeds was studied by \cite{agnetis2022,agnetis2025}. All these works focus on multiple-machine settings, since the single-machine case reduces to the classical $n$-out-of-$n$ test sequencing problem. Complexity arises in the multi-machine case due to job partitioning (for non-replicable jobs) or the fact that other sequences can become optimal (for replicable jobs). 

UJSSP, although single-machine, remains non-trivial because of the possibility to select jobs, with job-specific selection costs.
\cite{stadje1995} studied a closely related model in which a fixed number of jobs must be chosen, showing that a greedy algorithm is optimal in this case. Similar results for a more general model are found by \cite{chade2006}, who showed that a greedy algorithm is optimal for a more general selection problem with a downward-recursive revenue function and cardinality-dependent costs, encompassing both the equal-cost version of UJSSP and the problem of \cite{stadje1995}. \cite{olszewski2016} generalized their model further by allowing heterogeneous costs; they identified conditions under which the greedy algorithm still yields optimal solutions.



\section{Complexity}\label{sec:compl}
Before discussing the complexity of UJSSP in the general case, we discuss two special cases, namely with identical costs and with identical probabilities of success, which can both be solved by the greedy algorithm given by Algorithm~\ref{alg:ga}. This implies that both these special cases can be solved in $O(n^2)$, as at each step we simply need to ``try out" each of the remaining jobs and select one yielding the largest marginal expected gain. 

\begin{algorithm}
\caption{Greedy algorithm for UJSSP.} \label{alg:ga}
\begin{algorithmic}
\State $S_0 \gets \{\emptyset\}$;
\State $i \gets 0$;
\While{($i<n$ and) there is at least one item $j\not\in S_i$ such that $z(S_i\cup \{j\})>z(S_i)$}
\State Select the job $k \not\in S_i$ such that
    \[    z(S_i \cup \{k\}) = \max_{j \not\in S_i} \{z(S_i \cup \{j\})\};    \]
\State $S_{i + 1} \gets S_{i} \cup \{k\}$
\State $i \gets i+1$
\EndWhile
\State \Return $S_i$.
\end{algorithmic}
\label{algo1}
\end{algorithm}

\subsection{Identical costs}
Consider the special case of UJSSP where all jobs have the same cost, i.e., \(c_j = c\) for all \(j = 1,\dots,n\). This case is optimally solvable by a greedy algorithm, as implied by \cite{stadje1995} and \cite{chade2006}.

\cite{stadje1995} show that a variant of UJSSP with a fixed number of selected jobs admits a greedy optimum. Since the expected revenue function is submodular \citep{chade2006}, the marginal gain in expected reward decreases as the selected set grows. Therefore, in the equal-cost setting, no further job should be added once the marginal gain drops below the common cost \(c\).

\cite{chade2006} show that the greedy algorithm is optimal for a class of selection problems that includes UJSSP with equal costs. In particular, the \emph{College Selection Problem (CSP)}, introduced in their paper as an application, is equivalent to UJSSP\@. In CSP, a student selects a subset $S$ from a set $N$ of colleges ($|N|=n'$) to apply to. For each college $j$ there is a utility~$w_j$, a probability $\alpha_j$ of being accepted, and an application cost $c_j'$. The student applies to a number of colleges, and will pick one having the largest utility among those that admit her. If we number the colleges by non-increasing utilities, the probability that college $j$ is the one the student will actually go to when $S$ is selected, is given by \(\alpha_j\prod_{i<j, i\in S}(1-\alpha_i),\) and CSP consists in selecting a subset $S$ of colleges so as to maximize the expected utility, i.e., letting $c'(S)=\sum_{j\in S}c_j'$:
\begin{equation}\label{eq:college}
\max_{S\subseteq N}\left\{\sum_{j\in S}w_j\alpha_j\prod_{i<j, i\in S}(1-\alpha_i) -c'(S)\right\}.
\end{equation}
Given an instance of CSP, consider an instance of UJSSP with $n=n'$ jobs,  letting $r_j=w_j\alpha_j/(1-\alpha_j)$, $\pi_j=(1-\alpha_j)$ and $c_j=c_j'$. In this case \eqref{eq:netprofit} and \eqref{eq:college} coincide, and 
\[
Z_j=\frac{\pi_jr_j}{1-\pi_j}=w_j,
\]
so indeed the ordering based on \eqref{eq:zrule} of UJSSP coincides with the ordering by non-increasing utilities in CSP\@. Vice versa, one can transform an instance of UJSSP into an equivalent instance of CSP by letting $w_j=r_j\pi_j/(1-\pi_j)$, $\alpha_j=1-\pi_j$, and $c_j'=c_j$.
As shown by \cite{chade2006}, the CSP is optimally solved by a greedy algorithm when all application costs are equal. Consequently, UJSSP with equal costs admits the same optimal greedy solution.

\subsection{Identical probabilities}
In this section we consider the special case of UJSSP in which $\pi_j=\pi$ for all $j=1,\dots,n$. First of all, notice that in this case, letting $r_{[k]}$ denote the  reward of the job in position $k$, and denoting with $S$ the set of selected jobs, the expected reward can be written as: \(R(S)=\sum_{k=1}^{|S|} r_{[k]}\pi^k.\)
As a consequence, if a job $j$ is assigned to position $k$, its contribution to the net expected reward is given by
\begin{equation}\label{assign}
\max\{0,r_j\pi^k - c_j\},
\end{equation}
and for a fixed number $H$ of selected jobs, the optimal solution can be found by solving an assignment problem in which the reward $q_{jk}$ from assigning $j$ to position $k$ is given by \eqref{assign} if $k\leq H$, and 0 otherwise. Iterating for all values of $H$, we find the optimal solution of UJSSP.

A stronger result can be established by viewing UJSSP as a special case of the \emph{Simultaneous Selection Problem (SSP)}, introduced by  \cite{olszewski2016}. Consider a set $T$ of items, each having utility $u_i$, such that $u_1\geq u_2\geq\dots\geq u_n$, and a submodular non-decreasing set function~$f(\cdot)$. Then consider the following optimization problem: 
\begin{subequations}
	\begin{align}
	g(T)=\max &\sum_{i\in T}u_ix_i\label{eq:gt}\\
& \sum_{i\in S}x_i\leq f(S)\;\;\forall S\subseteq T\label{eq:sub}\\
& x_i\geq 0, i\in T\nonumber
	\end{align}	
\end{subequations}
This problem can be interpreted as the maximum utility that one can obtain by assigning weights~$x_i$ to the elements of $T$, given that for any subset $S$ of $T$ the sum of the weights does not exceed the submodular function $f(S)$. As $f(\cdot)$ is submodular, \eqref{eq:sub} is indeed a polymatroid and it is well known that $g(T)$ has the following optimal solution:
\begin{align}
    x_1^*&=f(\{1\})\label{eq:poly}\\    
    x_2^*&=f(\{1,2\})-f(\{1\})\nonumber\\
    \dots&\nonumber\\
    x_n^*&=f(\{1,2,3,\dots,n\})-f(\{1,2,3,\dots,n-1\}).\nonumber
\end{align} 
It is easy to check that UJSP (i.e., UJSSP without costs) is a special case of of this problem. In fact, given an instance of UJSP, in \eqref{eq:gt} let $u_j=Z_j$ and $f(S)=1-\prod_{j\in S}\pi_j$. From \eqref{eq:poly}, we get that the optimal solution is:
\begin{align}
    x_1^*&=1-\pi_1\nonumber\\    
    x_2^*&=(1-\pi_1\pi_2)-(1-\pi_1)=\pi_1(1-\pi_2)\nonumber\\
    \dots&\nonumber\\
    x_n^*&=(1-\pi_1\pi_2\pi_3\dots\pi_n)-(1-\pi_1\pi_2)=\pi_1\pi_2\dots\pi_{n-1}(1-\pi_n)\nonumber
\end{align} 
and hence \eqref{eq:gt} becomes 
\begin{align}
    g(T)=\sum_{i\in T}u_ix_i^*&=\frac{\pi_1r_1}{1-\pi_1}(1-\pi_1)\nonumber\\
    &+\frac{\pi_2r_2}{1-\pi_2}\pi_1(1-\pi_2)\nonumber\\    
    &+\dots\nonumber\\
    &+\frac{\pi_nr_n}{1-\pi_n}\pi_1\pi_2\dots\pi_{n-1}(1-\pi_n)\nonumber
\end{align}
which is identical to \eqref{eq:exprofit} (for $|S|=n$).
Let us now add the costs to the picture. Given a ground set $N$, and a cost $c_j$ for each variable $j=1,\dots,|N|$, for each $T \subseteq N$, let $z(T)=g(T)-c(T)$, where $g(T)$ is defined by \eqref{eq:gt} and $c(T)=\sum_{j\in T}c_j$.  The \textit{Simultaneous Selection Problem (SSP)} \citep{olszewski2016} consists in selecting a subset $T^*\subseteq N$ such that
\begin{equation*}
T^*=\arg\max_{T\subseteq N}\left\{z(T) \right\}.
\end{equation*}
From what is above, it is apparent that UJSSP is a special case of SSP.

The main result given in \cite{olszewski2016} is the following. 
\begin{theorem}\label{thm:down}
\citep{olszewski2016} Given an instance of SSP, if:
\begin{itemize}
    \item[$(i)$] $f(S)$ is downward recursive, and
    \item[$(ii)$] $f(S)$ only depends on the cardinality of $S$
\end{itemize}
then SSP is optimally solved by the greedy algorithm.
\end{theorem}
We note that a submodular function is \textit{downward recursive} if, for any two disjoint subsets $U$ and $V$ such that $\min_{j\in U}\{u_j\}\geq\max_{j\in V}\{u_j\}$, it holds
    \[
    f(U\cup V)=f(U)+f(V)\rho(U),
    \]
where $\rho(\cdot)$ is a multiplicative function (i.e., $\rho(U\cup V)=\rho(U)\rho(V)$) whose values belong to [0, 1).

We use Theorem~\ref{thm:down} to prove that UJSSP with identical probabilities can be solved by the greedy algorithm.
\begin{theorem}\label{thm:sspdr}
If all the jobs have the same success probability $\pi$,  UJSSP is solved by the greedy algorithm.
\end{theorem}
\begin{proof}
We only need to show that the two conditions $(i)$ and $(ii)$ of Theorem~\ref{thm:down} hold. 

$(i)$ From \eqref{eq:gt}--\eqref{eq:sub}, UJSSP can be  formulated as SSP, with $f(S)=1-\prod_{j\in S}\pi_j$. Consider the following function $\rho(S)=\prod_{j\in S}\pi_j$. Clearly, $\rho(S)$ is a multiplicative function, and given disjoint sets $U$ and $V$:
\[
f(U\cup V)=1-\prod_{i\in U\cup V}\pi_i=(1-\prod_{i\in U}\pi_i)+(\prod_{i\in U}\pi_i)(1-\prod_{i\in V}\pi_i)=f(U)+\rho(U)f(V),
\]
hence $f$ is downward recursive. 

$(ii)$ In the special case of UJSSP in which all the jobs have the same success probability $\pi$, for any subset $S$ one has \(f(S)=1-\pi^{|S|}\), hence $f(S)$ only depends on the cardinality of $S$. 
\end{proof}

\subsection{The general case}
SSP is NP-hard for a general submodular function $f$, even when $f(S)$ is of coverage-type \citep{feige2010}. However, we are not aware of results establishing the complexity of UJSSP\@. In the following, we address this issue by showing that UJSSP is in general NP-hard. In the proof we use the following problem, which is strongly NP-hard \citep{ng2010}.

\emph{Product Partition Problem (PPP)}. Given a set $N$ of $n$ positive integers $a_1,\dots,a_n$, is there a subset $S\subset N$ such that \(\prod_{j\in S}a_j=\prod_{j\in N\setminus S}a_j\)?
Note that we can assume $a_j\geq 2$ for all $j$.


\begin{theorem}\label{thm:comp}
UJSSP is NP-hard.
\end{theorem}

\begin{proof}
Given an instance of the PPP, we define an instance of UJSSP as follows, where we let $W=\prod_{j=1}^na_j$. There are $n$ jobs, such that:
\begin{equation}\label{eq:pij}
\pi_j:=\frac{1}{a_j} 
\end{equation}
\begin{equation}\label{eq:rj}
r_j:= \sqrt{W}\left(\frac{1-\pi_j}{\pi_j}\right)=\sqrt{W}(a_j-1)
\end{equation}
\begin{equation}\label{eq:cj}
c_j:=\ln{a_j}.
\end{equation}
We want to establish whether there is a subset of jobs that ensures a net expected reward of at least
\begin{equation}\label{target}
\sqrt{W}-1-\ln{\sqrt{W}}.
\end{equation}

From \eqref{eq:pij}--\eqref{eq:cj}, we observe that for any job $j$ one has 
$Z_j=\pi_jr_j/(1-\pi_j)=\sqrt{W}$, so the value of $Z_j$ is the same for all jobs. Hence, given $S$, the value of $R(S,\sigma)$ is indeed the same for any $\sigma$, given by:
\[
\pi_{\sigma(1)}r_{\sigma(1)}+\pi_{\sigma(1)}\pi_{\sigma(2)}r_{\sigma(2)}+\dots+\pi_{\sigma(1)}\pi_{\sigma(2)}\dots\pi_{\sigma(|S|)}r_{\sigma(|S|)}
\]
which, from \eqref{eq:rj} and after some algebra, yields (assume that $\prod_{j\in S}\pi_j=1$ if $S=\emptyset$):

\[
R(S)=\sqrt{W}(1-\prod_{j\in S}\pi_j).
\]
As a consequence, we can write the net expected profit as
\begin{equation}\label{zs}
z(S)=\sqrt{W}(1-\prod_{j\in S}\pi_j)-\sum_{j\in S}c_j.
\end{equation}

So, from \eqref{target} and \eqref{zs} we want to establish whether there is a subset $S$ of jobs such that \begin{equation}\label{target2}
\sum_{j\in S}c_j+\sqrt{W}\prod_{j\in S}\pi_j \leq 1+\ln{\sqrt{W}}.
\end{equation}

For any $S$, from \eqref{eq:pij} and \eqref{eq:cj} we can write the left-hand-side of \eqref{target2} as
\begin{equation*}
\ln{\prod_{j\in S}a_j}+\sqrt{W}\frac{1}{\prod_{j\in S}a_j}.
\end{equation*}
Now, the function
\[
f(x)=\ln{x}+\sqrt{W}\frac{1}{x}
\]
has a minimum for $x=\sqrt{W}$, and it is given by $1+\ln\sqrt{W}$. Hence, \eqref{target2} holds (at equality) if and only if there is a partition of the $n$ integers into two sets, each having product equal to $\sqrt{W}$.

We are left with showing that the reduction is polynomial. Denoting by $a_{\max}$ the largest integer appearing in the instance of PPP, the input size of such an instance is $O(n\log_2 a_{\max})$. The only issue is that the values of rewards and costs in the instance of UJSSP can be irrationals. Of course, precise encoding of an irrational would require an infinite number of bits, but we are only concerned with a precision level which allows distinguishing two different values of $\ln{\prod_{j\in S}a_j}$ for any $S$. As $\prod_{j\in S}a_j \leq W$, we only need enough bits to distinguish between $\ln W$ and $\ln (W+1)$. Since $\ln (W+1)-\ln W \geq \frac{1}{W}-\frac{1}{W^2}>\frac{1}{W^2}$ (for any $W\geq 3$), we need at most $2\lceil\log_2 W\rceil$ bits (corresponding to decimal digits) to distinguish between $\ln W$ and $\ln (W+1)$. On the other hand, as the smallest probability will be larger than $1/W$, each probability requires at most $\log_2W=O(n\log_2 a_{\max})$ bits each to be encoded, and the rewards at most $\log_2a_{\max}+1/2\log_2W=O(n\log_2 a_{\max})$ bits, as rewards are integer and each reward is smaller than $\sqrt{W}a_{\max}$. Finally, each cost can be encoded by means of $\log_2(\ln a_{\max})=O(\log_2a_{\max})$ bits for the integer part and, as recalled,  $2\lceil\log_2 W\rceil$ bits for the decimal part. Hence, as there are $n$ jobs, the total number of bits required to encode the instance of UJSSP is $O(n^2\log_2 a_{\max})$, which is polynomial in the input size of the instance of PPP. 
\end{proof}

Notice that even if PPP is strongly NP-hard, we only prove the ordinary NP-hardness of UJSSP, since in our proof we use the value $\sqrt{W}$, where $W$ is the product of all numbers of PPP, and such a product is not bounded by a polynomial in the maximum integer of PPP\@. In fact, the input of an instance of PPP does not include the number $\sqrt{W}$. 

We note here that the result in Theorem~\ref{thm:comp} implies that the SSP is NP-complete even when the submodular function $f(S)$ is downward recursive.

As the cases with identical costs or identical probabilities are solved by the greedy algorithm, one may be led to investigate if this is the case also when all rewards are identical. It turns out that the greedy algorithm may not provide the optimal solution in this case.

\begin{example}\label{ex:1}
Consider $n=3$ and the data in Table \ref{tab:controes}. In this example, all jobs have unit reward (hence the optimal sequence is by decreasing probabilities). At the first step, the greedy algorithm would select job 2, since $z({2})=\pi_2-c_2=0.1$ while $z({1})=\pi_1-c_1=0.09$ and $z({3})=\pi_3-c_3=0.09$. At the second step, the greedy algorithm adds job 1, as $z({1,2})=\pi_1+\pi_1\pi_2-c_1-c_2=0.103>z({2,3})=\pi_2+\pi_2\pi_3-c_2-c_3=0.1029$, and we see that  $\{1,2\}$ is indeed better than $\{2\}$. Subsequently adding also job 3 yields $z({1,2,3})=\pi_1+\pi_1\pi_2+\pi_1\pi_2\pi_3-c_1-c_2-c_3=0.0476<0.103$, so the greedy algorithm returns the set $\{1,2\}$. However, the optimal solution is $\{1,3\}$, with $z({1,3})=\pi_1+\pi_1\pi_3-c_1-c_3=0.113$.
\end{example}

\begin{table}[htbp]
\centering
\begin{tabular}{@{}c@{\quad}c@{\quad}c@{\quad}c@{}}
$j$ & $\pi_j$ & $c_j$ & $r_j$ \\ \hline
1 & 0.9 & 0.81 & 1\\
2 & 0.87 & 0.77 & 1\\
3 & 0.67 & 0.58 & 1 \\ \hline
\end{tabular}
\caption{Data for Example \ref{ex:1}.\label{tab:controes}}
\end{table}

The complexity of UJSSP in the special case of identical rewards is open. 

We give a partial result for the special case in which the product $\pi_jr_j$ (expected reward) is identical for all jobs $j$. In fact, the following lemma holds.
\begin{lemma} \label{lem:eqer}
Let \( J \) be a set of \( n \) jobs such that for all \( j \in J \): \( r_j \pi_j = k \), for some \( k > 0 \). A subset \( S = \{j_1, \ldots, j_{|S|}\} \subseteq J \), can only be optimal if:
\begin{enumerate} 
    \item The job with the lowest cost is included in \( S \), and
    \item For every \( i \in \{1, \ldots, |S|-1\} \), the job \( j = \arg\min_{j > j_i} c_j \) is also included in \( S \).
\end{enumerate}
\end{lemma}
\begin{proof}
The expected profit of executing \( S \) in order is:
\[
z(S) = \sum_{i=1}^{|S|} \left( k \cdot \prod_{l=1}^{i-1} \pi_{j_l} \right) - \sum_{i=1}^{|S|} c_{j_i}.
\]
Suppose \( j_{\min} = \arg\min_{j \in J} c_j\) is not in \(S\). Let \( S' = (S \setminus \{j_{|S|}\}) \cup \{j_{\min}\} \). Then:
\[
z(S') \geq z(S) + (c_{j_{|S|}} - c_{j_{\min}}) > z(S),
\]
since replacing the last job in the sequence by \(j_{\min}\) and keeping the order the same does not influence the revenue component and increases the expected profit with \(c_{j_{|S|}} - c_{j_{\min}} > 0\), which contradicts optimality of \(S\).
Similarly, for any \( i \in \{1,\ldots,|S|-1\} \), let \( j = \arg\min_{j > j_i} c_j \). If \( j \notin S \), define \( S' \) by replacing the last element of \( S \) with \( j \). The profit again increases due to the cost reduction, contradicting optimality of \(S\).
\end{proof}
 
Lemma~\ref{lem:eqer} implies that in this special case, the greedy algorithm is optimal if $|S^*|=2$ or if $n\leq 3$. However, for \(n = 4\) we can already find an example for which the greedy algorithm does not find the optimal solution.
\begin{example}\label{ex:2}
Consider $n=4$ and the data in Table \ref{tab:controes2}. In this example, for all the jobs the product of reward and probability is equal to 80 (hence the optimal sequence is by decreasing probabilities). At the first step, the greedy algorithm would select job 3 since it has the lowest cost: $z({3}) = 80 - 8 = 72$. At the second step, the greedy algorithm adds job 2, as $z({2,3})=80+0.4\cdot80-24-8=80>z({1,3})=80 + 0.8\cdot80-57-8=79=z({3,4})=80+0.2 \cdot 80-8-9=79>z({3})$. In a third step, the greedy algorithm adds job 1 because $z({1,2,3})=80+0.8\cdot80+0.8\cdot0.4\cdot80-57-24-8 = 80.6>z({2,3}) > z({2,3,4})=80+0.4\cdot80+0.4\cdot0.2\cdot80-57-8-9=77.4$. The greedy algorithm stops here as adding job~4 to $\{1,2,3\}$ does not lead to an improvement. However, in the optimal solution we select $\{1,3,4\}$, leading to an objective function value of $z({1,3,4}) = 80 + 0.8\cdot80 + 0.8\cdot0.2\cdot80 - 57 - 8 - 9 = 82.8$.
\end{example}

\begin{table}[htbp]
\centering
\begin{tabular}{@{}c@{\quad}c@{\quad}c@{\quad}c@{}} 
$j$ & $\pi_j$ & $c_j$ & $r_j$ \\ \hline
1 & 0.8 & 57 & 100\\
2 & 0.4 & 24 & 200\\
3 & 0.2 & 8 & 400\\ 
 4 & 0.1 & 9 & 800 \\ \hline
\end{tabular}
\caption{Data for Example \ref{ex:2}.\label{tab:controes2}}
\end{table}

\section{A MILP formulation}\label{sec:milp}
Exploiting the ordering induced by \eqref{eq:zrule}, it is possible to give a compact MILP formulation for UJSSP\@. In fact, numbering the jobs by non-increasing values of $Z_j=\pi_jr_j/(1-\pi_j)$, we can exploit the Z-ordering to avoid using disjunctive constraints that are  typical of many compact formulations for scheduling problems.
\begin{subequations}
	\begin{align}
	\max & \sum_{j=1}^n (r_jP_j-c_jx_j) & \label{mip:fobj}\\
	& P_j \leq  \pi_jx_j & \forall j \in \{1,\dots,n\} \label{mip:initial'}\\	
    & P_j \leq \pi_j(P_i+1-x_i) & \forall i,j \in \{1,\dots,n\} : i<j \label{mip:po_order'}\\
	 &  x_{j} \in \{0,1\} & \forall j \in \{1,\dots,n\} \label{mip:x_bin'}\\
 &P_{j} \geq 0 & \forall\, j \in \{1,\dots,n\} \label{mip:P_real}
	\end{align}	
\end{subequations}
In this formulation, $x_j=1$ if job $j$ is selected. If $x_j=1$, constraints \eqref{mip:initial'} and \eqref{mip:po_order'} ensure that $P_j$ is the product of the probabilities of the jobs selected up to $j$ (included), whereas, if and only if job $j$ is not selected, $P_j=0$.  Note that all the constraints \eqref{mip:po_order'} containing $x_i$ become non-binding if $x_i=0$.

\section{An upper bound}\label{sec:ub}

In what follows, we develop an upper bound for the problem. Recall that the jobs are numbered in non-increasing order of $Z_j$. We denote by $\pi^{[s]}_{1:l}$ the $s$-th largest probability in $\{\pi_1,\pi_2,\dots,\pi_l\}$, $s=1,\dots,l$ (and we let $\pi^{[0]}_{1:l}=1$). Now consider a job $j$, and suppose that it is scheduled in position $k\leq j$. An upper bound on its contribution to the objective function is given by
\begin{equation}
\max\{0,\left(\Pi_{s=0}^{k-1}\pi^{[s]}_{1:j-1}\right)\pi_jr_j-c_j\}.
\end{equation}
We now define the following quantities:
\[
q_{jk}=\left\{\begin{array}{cc}
 \max\{0,\left(\Pi_{s=0}^{k-1}\pi^{[s]}_{1:j-1}\right)\pi_jr_j-c_j\}    &  \text{if  } k\leq j\\
  0   & \text{if  } k> j
\end{array}\right.
\]
An upper bound can be derived by solving the following assignment problem, in which $x_{jk}=1$ if job $j$ is assigned to position $k$:
\begin{subequations}
	\begin{align}
	& \text{max} & \sum_{j =1}^n\sum_{k=1}^n q_{jk} x_{jk}  \label{ass:obj}\\
	& & \sum_{k=1}^n x_{jk} & =  1 & & \text{$\forall\, j=1,\dots,n$} \label{ass:1}\\	&&\sum_{j=1}^n x_{jk} & =  1 & & \text{$\forall\, k=1,\dots,n$} \label{ass:2}\\
    &&x_{jk}&\geq 0& & \text{$\forall\, j=1,\dots,n, k=1,\dots,n$}
\end{align}\end{subequations}
This upper bound can be slightly refined. It can happen that a job $j$ is assigned to a position $k$, while a job $j'>j$ is assigned to position $k'<k$. We can make sure the Z-order is maintained by adding constraints \eqref{uncr:1} or constraints \eqref{uncr:2.1}, \eqref{uncr:2.2}, \eqref{uncr:2.3}, and \eqref{uncr:2.4}.
\begin{subequations}
    \begin{align}
    && x_{jk} + x_{j'k'} \leq 1 && \text{$\forall\, j,j',k,k' \in \{1,\dots,n\}: j < j', k' < k$} \label{uncr:1} \\
   && \sum_{j' < j} \sum_{k' > k} x_{j'k'} \leq M_{1jk} \cdot \delta_{1jk} && \forall j \in \{2,\ldots,n\}, k \in \{1,\ldots,n-1\} \label{uncr:2.1} \\
   && \sum_{j' \geq j} \sum_{k' \leq k} x_{j'k'} \leq M_{2jk} \cdot \delta_{2jk} && \forall j \in \{2,\ldots,n\}, k \in \{1,\ldots,n-1\} \label{uncr:2.2} \\
   && \delta_{1jk} + \delta_{2jk} \leq 1 && \forall j \in \{2,\ldots,n\}, k \in \{1,\ldots,n-1\} \label{uncr:2.3} \\
   && \delta_{1jk}, \delta_{2jk} \in \{0,1\} && \forall j \in \{2,\ldots,n\}, k \in \{1,\ldots,n-1\} \label{uncr:2.4}
\end{align}
\end{subequations}
with \(M_{1jk} = \min \{n-k,j\}\) and \(M_{2jk} = \min \{k,n-j\}\).

\section{Dynamic programming} \label{sec:dp}
If all costs $c_j$ are integer, we can solve UJSSP using a backward dynamic programming (DP) algorithm. Let the state be defined by a pair $(b,i)$, where $b$ denotes the remaining budget and $i$ denotes that jobs from the set $\{i, \ldots, n\}$ can be selected. 
The value function $g(b,i)$ represents the maximum attainable expected revenue under these constraints.  
The recurrence is given by
\begin{equation*}
g(b,n) =
\begin{cases}
0 & \text{if } b \in \{0, \ldots, c_n - 1\}, \\
\pi_n \cdot r_n & \text{if } b \in \{c_n, \ldots, \sum_{j \in J} c_j\},
\end{cases}
\end{equation*}
and for $i < n$:
\begin{equation*}
g(b,i) = 
\begin{cases}
g(b,i+1) & \text{if } b \in \{0, \ldots, c_i - 1\}, \\
\max \left\{ g(b,i+1), \ \pi_i \cdot \big( r_i + g(b - c_i, i+1) \big) \right\} & \text{if } b \in \{c_i, \ldots, \sum_{j \in J} c_j\}.
\end{cases}
\end{equation*}

Intuitively, at each decision point we choose between skipping item $i$ (thus retaining budget $b$ and moving to the next index) 
or selecting item $i$ (which consumes $c_i$ units of budget, yields reward $\pi_i \cdot r_i$, and continues optimally with the remaining budget).  

The algorithm proceeds by initializing all values $g(b,n)$ for each budget $b$, 
and then iteratively computing values for 
$g(b,n-1), g(b,n-2), \ldots, g(b,1)$. 
Since computations for different budgets are independent given a certain set of jobs $\{i, \ldots, n\}$, each of these computations could, in principle, be calculated in parallel.  

Finally, the optimal solution is obtained by evaluating
\begin{equation*}
\max_{b} \, \{ g(b,1) - b \},
\end{equation*}
which selects the budget level that maximizes the objective function value.  
The algorithm runs in $O(n \cdot \sum_{j\in J} c_j )$ time, which can be upper-bounded by $O(n^2 \cdot c_{\max})$. Because the DP algorithm runs in pseudopolynomial time, we can refine the complexity classification of the problem: the version with integer costs is weakly NP-hard. This refinement does not extend to instances with fractional costs.

\section{Stepwise exact algorithms}\label{sec:ex}

We present two novel exact algorithms for solving UJSSP, which proceed stepwise through the jobs, either forwards or backwards. While they share several common ideas, their details differ.

\subsection{Foundations of the stepwise exact algorithms}
As mentioned before, we assume that the jobs are indexed in non-decreasing order of $Z_j$. For any set \( S \subseteq J \) and any \( j \in \{1, \ldots, n\} \), we define:
\[
\begin{aligned}
S_{<j} &= S \cap \{1, \ldots, j-1\}, &\quad S_{\leq j} &= S \cap \{1, \ldots, j\}, \\
S_{>j} &= S \cap \{j+1, \ldots, n\}, &\quad S_{\geq j} &= S \cap \{j, \ldots, n\}.
\end{aligned}
\]
In particular, $J_{<j} = \{1, \ldots, j-1\}$, $J_{\leq j} = \{1, \ldots, j\}$, $J_{>j} = \{j+1, \ldots, n\}$, and $J_{\geq j} = \{j, \ldots, n\}$. Let $S^* \subseteq J$ denote an optimal solution to UJSSP; then \( S^*_{<j} \), \( S^*_{\leq j} \), \( S^*_{>j} \), and \( S^*_{\geq j} \) are defined analogously.

\paragraph{Optimality conditions}
For any $S \subseteq J$ and any index $j$, the objective function can be rewritten as:
\[
R(S_{\leq j}) + \left( \prod_{i \in S_{\leq j}} \pi_i \right) \cdot R(S_{>j}) - c(S_{\leq j}) - c(S_{>j}).
\]
An optimal solution \( S^* \) maximizes this expression over all subsets \(S\) of \(J\). Therefore, if \( S^*_{>j} \) is given, then an optimal subset \( S^*_{\leq j} \subseteq J_{\leq j} \) must maximize:
\[
R(S_{\leq j}) + \left( \prod_{i \in S_{\leq j}} \pi_i \right) \cdot R(S^*_{>j}) - c(S_{\leq j})
\]
over all subsets \( S_{\leq j} \subseteq J_{\leq j} \), where \( c(S^*_{>j}) \) is constant (for any fixed $S^*_{>j}$) and can be ignored when comparing subsets. Similarly, if \( S^*_{\leq j} \) is given, then the optimal following jobs \( S^*_{>j} \subseteq J_{>j} \) must maximize:
\[
\left( \prod_{i \in S^*_{\leq j}} \pi_i \right) \cdot R(S_{>j}) - c(S_{>j})
\]
over all \( S_{>j} \subseteq J_{>j} \), now ignoring the constant terms \( R(S^*_{\leq j}) \) and \( c(S^*_{\leq j}) \).

\paragraph{Parameterizing the unknowns}
In practice, we do not know \( R(S^*_{>j}) \) or \( \prod_{i \in S^*_{<j}} \pi_i \). We therefore treat these quantities as unknown variables in the algorithm. Let $
r$ stand for $R(S^*_{>j})$ and $p$ represent $\prod_{i \in S^*_{<j}} \pi_i$. Define the following parameterized profit functions:
\[
F(S_{\leq j}, r) := R(S_{\leq j}) - \sum_{i \in S_{\leq j}} c_i + \left( \prod_{i \in S_{\leq j}} \pi_i \right) \cdot r,
\]
\[
B(S_{\geq j}, p) := - \sum_{i \in S_{\geq j}} c_i + R(S_{\geq j}) \cdot p.
\]
Then $S_{\leq j}$ can only equal $S^*_{\leq j}$ if there exists a feasible value \( r \) such that \( F(S_{\leq j}, r) \geq F(S'_{\leq j}, r) \) for all \( S'_{\leq j} \subseteq J_{\leq j} \) and \( S_{\geq j} \) can only equal \( S^*_{\geq j} \) if there exists a feasible value \( p \) such that \( B(S_{\geq j}, p) \geq B(S'_{\geq j}, p) \) for all \( S'_{\geq j} \subseteq J_{\geq j} \).

\paragraph{Upper envelope computation}
Both \( F(\cdot; r) \) and \( B(\cdot; p) \) are affine in the unknown parameter \( r \) or \( p \), respectively. Hence, finding all candidate subsets reduces to computing the upper envelope of a family of affine functions. Finding the upper envelope of lines of the form \( y = a_i \cdot x + b_i \) is equivalent to computing the lower convex hull of the points \( (a_i, -b_i) \), as noted by \cite{berg2000}. Finding the convex hull of a set of points can be done efficiently (see, amongst others, \cite{graham1972}, \cite{andrew1979}, \cite{jarvis1973}) even if points are added sequentially \citep{preparata1979}.

\paragraph{Bounding the unknowns}
Although \( r \) and \( p \) are not known precisely, their feasible values can be bounded:
\begin{itemize}
    \item \( 0 \leq r = R(S^*_{>j}) \leq R(J_{>j}) \), with $r=0$ at \( j = n \),
    \item \( \prod_{i<j} \pi_i \leq p = \prod_{i \in S^*_{<j}} \pi_i \leq 1 \), with $p=1$ at \( j = 1 \).
\end{itemize}

\paragraph{Pruning conditions}
If a subset $S_{\leq j}$ is suboptimal (i.e., cannot equal $S^*_{\leq j}$ for any feasible $r$), then any of its extensions—such as $S_{\leq j} \cup \{j+1\}$—can also be pruned. Similarly, suboptimal $S_{\geq j}$ cannot be extended into $S^*_{\geq j-1}$.

\paragraph{Implications}
These structural properties form the foundation of our exact algorithms. They enable efficient pruning strategies that substantially reduce the search space.
Although our approach incrementally builds an optimal solution by evaluating partial subsets, it cannot be formally classified as a DP method. This is because the subproblems we consider are not independent in the sense required by DP: the fact that a subset \(S \subseteq J_{\leq j}\) is optimal if the total job set were \(J_{\leq j}\) does not guarantee that it can be extended to form a globally optimal solution \(S^*\), because the probabilities of success of a predecessor set are multiplied with the revenue but not with the cost of a successor set. Similarly, our approach does not fully conform to a traditional branch-and-bound framework. While each node in our search corresponds to a binary decision, namely adding or excluding a job, and certain nodes are pruned based on their maximum achievable objective function value, the algorithm does not bound the objective function value in a traditional way and does not allow for a typical depth-first or best-first exploration strategy. Nonetheless, the framework could be adapted to a more conventional branch-and-bound scheme by defining the upper bound at every node as the highest attainable value of the profit function (i.e., taking the maximum value for \(r\) or \(p\) into account) and the lower bound as the corresponding minimum. Such an adaptation would permit depth-first or best-first search, but would not exploit the full information embedded in the profit functions.

\subsection{Forward stepwise exact algorithm}
In the forward exact solution method, we construct an optimal set \( S^* \) incrementally, starting from \( S^*_{\leq 1} \) and progressing toward \( S^*_{\leq n} = S^* \). At each step we try to eliminate subsets from \(J_{\leq j}\) as candidates for \( S^*_{\leq j} \). In doing so, we take advantage of the fact that even if $R(S^*_{> j})$ is unknown, it is upper bounded by $R(J_{> j})$, which can be easily computed, and decreases as \(j\) grows.

Let \(\mathcal{S}^{c}_{\leq j}\) be the set of candidates for \(S^*_{\leq j}\) before evaluating the profit function \(F(S,r)\) and let \(\mathcal{S}^{c^*}_{\leq j}\) denote the set of remaining candidates after taking \(F(S,r)\) into account. The set \(\mathcal{S}^{c}_{\leq j+1}\) of candidates for \(S^*_{\leq j+1}\) before the evaluation of their profit functions is obtained by duplicating each set in \(\mathcal{S}^{c^*}_{\leq j}\) and by appending job $j+1$ to the duplicate set so that \(|\mathcal{S}^{c}_{\leq j+1}| = 2 \cdot |\mathcal{S}^{c^*}_{\leq j}|\). However, the comparison among profit functions allows a significant reduction in the number of candidates, so that the size of sets \(\mathcal{S}^{c^*}_{\leq j}\) remains viable from a computational viewpoint. When the last job $n$ is reached, a simple comparison among the surviving sets points out the optimal solution.

The pseudocode for the forward exact algorithm is presented in Algorithm~\ref{alg:exf}. At each step, the set  \(S^{c^*}_{\leq j}\) is built by including nondominated sets of \(S^{c}_{\leq j}\) one by one. A detailed example illustrating the execution of the algorithm is provided below, using the data given in Table \ref{tab:exexemp}.

\begin{algorithm}[htbp]
    \caption{Forward stepwise exact algorithm} \label{alg:exf}
    \begin{algorithmic}
        \State $\bar{r} \gets R(J)$ 
        \State $\mathcal{S}^{c^*}_{\leq 0} \gets \{\emptyset\}$
        \For{$j \in \{1,\ldots,n\}$}
            \State $\mathcal{S}^{c}_{\leq j} \gets \mathcal{S}^{c^*}_{\leq j-1}$ 
            \For{$S \in \mathcal{S}^{c^*}_{\leq j-1}$}
            \If{$F(S \cup \{j\},r)$ does not coincide with some $F(S',r)$, $S' \in \mathcal{S}^{c}_{\leq j}$}
                \State $\mathcal{S}^{c}_{\leq j} \gets \mathcal{S}^{c}_{\leq j} \cup \{S \cup \{j\}\}$ 
            \EndIf
            \EndFor
            \State $\bar{r} \gets \frac{\bar{r} - \pi_j \cdot r_j}{\pi_j}$
            \State $\mathcal{S}^{c^*}_{\leq j} \gets \emptyset $
            \For{$S \in \mathcal{S}^{c}_{\leq j}$}
                \If{$\exists r \in [0, \bar{r}] \mid \forall S' \in \mathcal{S}^{c}_{\leq j} \setminus S: F(S,r) > F(S',r)$}
                    \State $\mathcal{S}^{c^*}_{\leq j} \gets \mathcal{S}^{c^*}_{\leq j} \cup \{ S \}$
                \EndIf
            \EndFor
        \EndFor
        \State \Return $\mathcal{S}^{c^*}_{\leq n}$ (optimal solution)
    \end{algorithmic}
\end{algorithm}

\begin{table}[htbp]
\centering
\begin{tabular}{@{}c@{\quad}c@{\quad}c@{\quad}c@{\quad}c@{}}
$j$ & $\pi_j$ & $c_j$ & $r_j$ & $Z_j$ \\ \hline 
1 & 0.75 & 75 & 250 & 750 \\
    2 & 0.5  & 150 & 500 & 500 \\
    3 & 0.5  & 70 & 350 & 350 \\
    4 & 0.6  & 30 & 100 & 150 \\ \hline
\end{tabular}
\caption{Input data for exact stepwise algorithms.\label{tab:exexemp}}
\end{table}

\paragraph{Step 1 (\( j = 1 \))}
The candidate subsets at this stage are \(\mathcal{S}^{c}_{\leq 1} = \{ \emptyset , \{1\} \}\). Their associated profit functions are:
\begin{align*}
F(\emptyset, r) &= r \\
F(\{1\}, r) &= 250 \cdot 0.75 - 75 + 0.75 \cdot r = 112.5 + 0.75 \cdot r
\end{align*}
To compare these two sets, we compute the upper bound on the optimal revenue for $S^*_{>1}$, i.e., $R(\{2,3,4\})$:
\[
0.5 \cdot 500 + 0.25 \cdot 350 + 0.15 \cdot 100 = 250 + 87.5 + 15 = 352.5.
\]
As shown in Figure~\ref{fig:forward-step1}, the function \(F(\{1\},r)\) lies strictly above \(F(\emptyset,r)\) for all feasible values of~\(r\). Therefore, \(\{1\}\) dominates \(\emptyset\), and we conclude that job 1 must be included in any optimal set. Hence, the reduced candidate set becomes \(\mathcal{S}^{c^*}_{\leq 1} = \{1\}\).

 \begin{figure}[htbp]
 \centering
 \includegraphics{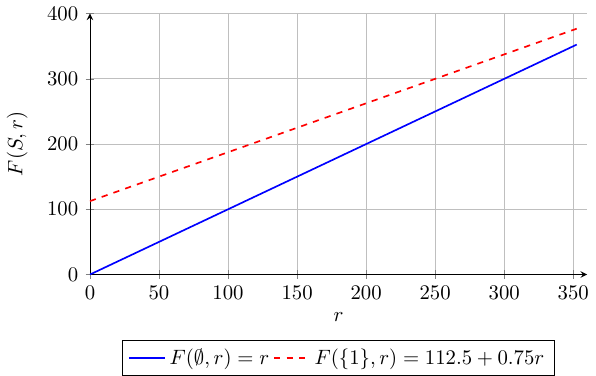}
 \caption{Comparison of sets to get \( \mathcal{S}^{c^*}_{\leq 1} \). \label{fig:forward-step1}}
\end{figure}

\paragraph{Step 2 (\( j = 2 \)).}
Since job 1 is guaranteed to be part of the optimal solution, we only consider subsets that include it: \( \mathcal{S}^{c}_{\leq 2} = \{ \{1\}, \{1,2\} \} \). The corresponding profit functions are:
\begin{align*}
F(\{1\}, r) &= 112.5 + 0.75 \cdot r \\
F(\{1,2\}, r) &= 250 \cdot 0.75 + 500 \cdot 0.75 \cdot 0.5 - (75 + 150) + (0.75 \cdot 0.5) \cdot r \\
&=150 + 0.375 \cdot r
\end{align*}
To evaluate which set yields a higher expected profit, we again compute the upper bound on the remaining revenue:
\[
0.5 \cdot 350 + 0.3 \cdot 100 = 175 + 30 = 205.
\]
Now, when we compare the two profit functions, we get:
\[
112.5 + 0.75 \cdot r > 150 + 0.375 \cdot r \iff r > 100.
\]
Therefore \( \{1,2\} \) dominates \( \{1\} \) if \( r < 100 \) while \( \{1\} \) dominates \( \{1,2\} \)  if \(r > 100\), as illustrated in Figure~\ref{fig:forward-step2}. Hence, both sets must be retained as candidates for \(S^*_{\leq 2}\), so \(\mathcal{S}^{c^*}_{\leq 2} = \{\{1\},\{1,2\}\}\).

\begin{figure}[htbp]
\centering
\includegraphics{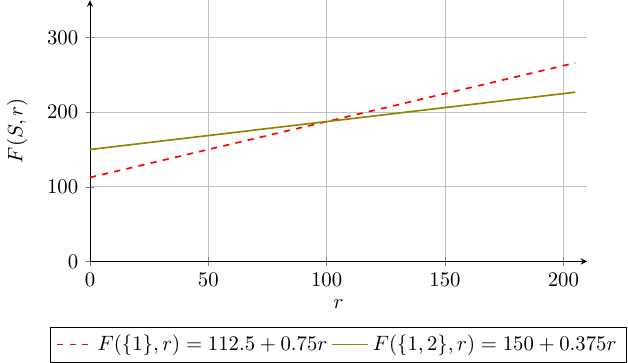} 
\caption{Comparison of sets to get \( \mathcal{S}^{c^*}_{\leq 2} \). \label{fig:forward-step2}}{}
\end{figure}

\paragraph{Step 3 (\( j = 3 \)).}
Again, to construct \(\mathcal{S}^{c}_{\leq 3}\) we take all sets in \(\mathcal{S}^{c^*}_{\leq 2}\) and their duplicates with the addition of job 3: \(\mathcal{S}^{c}_{\leq 3} = \{ \{1\}, \{1,2\}, \{1,3\}, \{1,2,3\} \}\).  The profit functions are as follows (\(F(\{1\},r)\) and \(F(\{1,2\},r)\) were derived in Step 2).

\begin{align*}
F(\{1,3\},r) &= (250 \cdot 0.75 + 350 \cdot 0.75 \cdot 0.5) - (75 + 70) + 0.75 \cdot 0.5 \cdot r \\
&= 173.75 + 0.375 \cdot r \\
F(\{1,2,3\},r) &= (250 \cdot 0.75 + 500 \cdot 0.75 \cdot 0.5 + 350 \cdot 0.75 \cdot 0.5 \cdot 0.5)\\
&\quad- (75 + 150 + 70) + 0.75 \cdot 0.5 \cdot 0.5 \cdot r \\
&= 145.625 + 0.1875 \cdot r
\end{align*}

The upper bound on the remaining revenue is \( 0.6 \cdot 100 = 60 \). As shown in Figure~\ref{fig:forward-step3}, for all feasible values of \(r \in [0,60]\) the set \(\{1,3\}\) achieves the highest profit. Therefore, we retain only this set:
\[
\mathcal{S}^{c^*}_{\leq 3} = \{ \{1,3\} \} = \{ S^*_{\leq 3} \}.
\]

\begin{figure}[htbp]
\centering
 \includegraphics{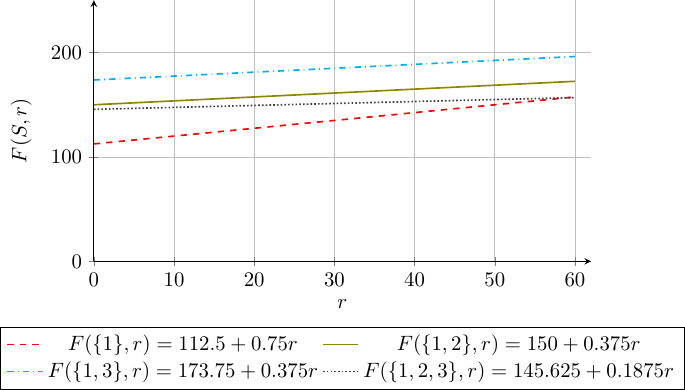} 
 \caption{Comparison of sets to get \( \mathcal{S}^{c^*}_{\leq 3} \). \label{fig:forward-step3}}
\end{figure}

\paragraph{Step 4 (\( j = 4 \)).}
At this point, $\mathcal{S}^{c}_{\leq 4}$ contains two candidate sets for \(S_{\leq 4}^*=S^*\) before comparing profit functions, namely \( \{1,3\} \) and \( \{1,3,4\} \). Since the remaining revenue is zero (\( R(S^*_{>4}) = 0 \)), we can simply evaluate the objective function values for both sets.
\begin{align*}
z(\{1,3\}) &= F(\{1,3\},R=0) = 173.75 \\
z(\{1,3,4\}) &= F(\{1,3,4\},R=0) \\
&= 250 \cdot 0.75 + 350 \cdot 0.75 \cdot 0.5 + 100 \cdot 0.75 \cdot 0.5 \cdot 0.6 - (75 + 70 + 30)\\
&= 166.25.
\end{align*}
Since \(z(\{1,3\}) > z(\{1,3,4\})\) we conclude that the optimal solution is \( S^* = \{1,3\} \). Notice that to obtain this solution, only six subsets needed to be evaluated:
\[
\emptyset, \{1\}, \{1,2\}, \{1,3\}, \{1,2,3\},\text{ and } \{1,3,4\},
\]
which is a significant reduction compared to the full set of \( 2^4 = 16 \) possible subsets of \( J \).

\subsection{Backward stepwise exact algorithm}
The backward algorithm follows a symmetric logic to the forward algorithm, but  
it starts from \( S^*_{\geq n} \) and moves backwards toward \( S^*_{\geq 1} = S^* \). At each step, we exclude subsets from \(J_{\geq j}\) as candidates for \( S^*_{\geq j} \) based on the profit function $B(\cdot,p)$, which depends on the joint probability of success for jobs scheduled \textit{before} job $j$, i.e., the earliest job in the set, and which can be lower-bounded by the product of all probabilities in $J_{<j}$, which increases as we move backwards through the job sequence. Let \(\mathcal{S}^{c}_{\geq j}\) be a set of candidate sets for \(S^*_{\geq j}\). If, for some $S\in \mathcal{S}^{c}_{\geq j}$, one has that for each $p \geq \prod_{i \in J_{< j}} \pi_i$, value $B(S,p)$ is not larger than 
all other profit functions in $\mathcal{S}^{c}_{\geq j}$, then $S$ can be discarded from 
$S^*_{\geq j}$.

Let $\mathcal{S}^{c}_{\geq j}$ be the set of candidates for $S^*_{\geq j}$ before taking \(B(S,p)\) into account and let \(\mathcal{S}^{c^*}_{\geq j}\) be the filtered set after suboptimal subsets are eliminated using the profit function \(B(S,p)\). The set $\mathcal{S}^{c}_{\geq j-1}$ is obtained from \(\mathcal{S}^{c^*}_{\geq j}\) by duplicating its members and adding job $j-1$ to each duplicate set. When job 1 is finally reached, a simple comparison among the surviving sets gives the optimal solution. The pseudocode for the backward exact algorithm is given in Algorithm~\ref{alg:exb}. The steps for each job~\(j\) using the data in Table~\ref{tab:exexemp} are worked out below as illustration.

\begin{algorithm}[htbp]
    \caption{Backward stepwise exact algorithm} \label{alg:exb}
    \begin{algorithmic}
        \State $\ubar{p} \gets \prod_{j \in J} \pi_j$ 
        \State $\mathcal{S}^{c^*}_{\geq n + 1} \gets \{\emptyset\}$
        \For{$j \in \{n,\ldots,1\}$}
            \State $\mathcal{S}^{c}_{\geq j} \gets \mathcal{S}^{c^*}_{\geq j+1}$
            \For{$S \in \mathcal{S}^{c^*}_{\geq j+1}$}
            \If{$B(S \cup \{j\},p)$ does not coincide with some $B(S',p)$, $S' \in \mathcal{S}^{c}_{\geq j}$}
                \State $\mathcal{S}^{c}_{\geq j} \gets \mathcal{S}^{c}_{\geq j} \cup \{S \cup \{j\}\}$
            \EndIf
            \EndFor
            \State $\ubar{p} \gets \frac{\ubar{p}}{\pi_j}$ 
            \State $\mathcal{S}^{c^*}_{\geq j} \gets \emptyset$
            \For{$S \in \mathcal{S}^{c}_{\geq j}$}
                \If{$\exists p \in [\ubar{p}, 1] \mid \forall S' \in \mathcal{S}^{c}_{\geq j} \setminus S : B(S,p) > B(S',p)$}
                    \State $\mathcal{S}^{c^*}_{\geq j} \gets \mathcal{S}^{c^*}_{\geq j} \cup \{ S \}$
                \EndIf
            \EndFor
        \EndFor
        \State \Return $\mathcal{S}^{c^*}_{\geq 1}$ (optimal solution)
    \end{algorithmic}
\end{algorithm}

\paragraph{Step 1 (\( j = 4 \)).}
We start with the candidate subsets \(S^{c}_{\geq 4} = \{\emptyset,\{4\}\}\). The corresponding profit functions are:
\begin{align*}
B(\emptyset, p) &= 0 \\
B(\{4\}, p) &= - 30 + p \cdot 0.6 \cdot 100 = - 30 + p \cdot 60.
\end{align*}
The lower bound on the joint success probability of earlier jobs is given by:
\[
\prod_{i \in J_{<4}} \pi_i=0.75 \cdot 0.5 \cdot 0.5 = 0.1875.
\]
To determine when \(\{4\}\) yields higher profits than \(\emptyset\), we consider:
\[
-30 + p \cdot 60 > 0 \iff p > 0.5.
\]
Since \(p \in [0.1875,1]\), both options need to be retained as candidates for \(S^*_{\geq 4}\) as the set \( \{4\} \) is better for \( p > 0.5 \) and \( \emptyset \) is better for \(p < 0.5\). 

 \begin{figure}[htbp]
 \centering
\includegraphics{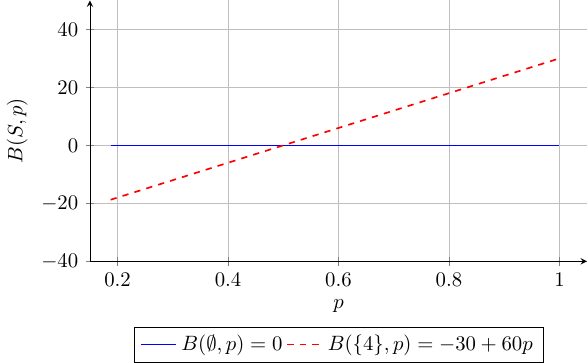}
\caption{Comparison of sets to get \( S^{c^*}_{\geq 4} \). \label{fig:backward-step1}}
\end{figure}

\paragraph{Step 2 (\( j = 3 \)).}
To get to the new candidate set we consider the sets in \(S^{c^*}_{\geq 4}\) and their extensions with job 3: \( S^{c}_{\geq 3} = \{ \emptyset, \{4\}, \{3\}, \{3,4\} \} \). We compute the profit functions for the sets involving job~3 (the others have been computed in the previous step):
\begin{align*}
B(\{3\}, p) &= -70 + 0.5 \cdot 350 \cdot p\\
&= - 70 + 175 \cdot p \\
B(\{3,4\}, p) &= - 70 - 30 + (0.5 \cdot 350 + 0.5 \cdot 0.6 \cdot 100) \cdot p  \\
&= -100 + 205 \cdot p
\end{align*}
The updated lower bound on the joint success probability before job 3 is \(0.75 \cdot 0.5 = 0.375 \). We now analyze the dominance relationships between the candidate sets, as visualized in Figure~\ref{fig:backward-step2}. We see that \(\{3\}\) dominates both \(\{4\}\) and \(\{3,4\}\) for all feasible values of $p$, with equality at \(p=1\) for \(\{3,4\}\). The empty set outperforms \(\{3\}\) for \(p < 0.4\). In conclusion, \(S^{c^*}_{\geq 3} = \{ \emptyset, \{3\} \} \). 

 \begin{figure}[htbp]
 \centering
\includegraphics{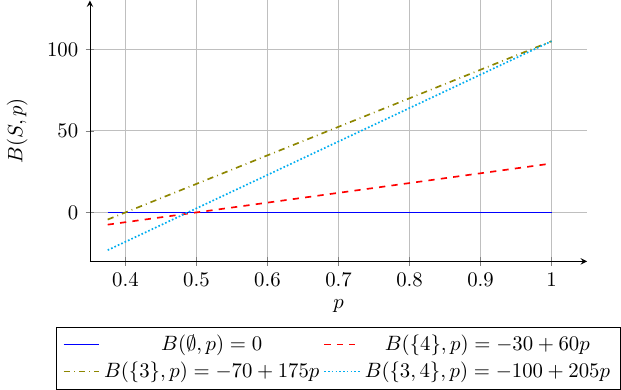}
\caption{Comparison of sets to get \( S^{c^*}_{\geq 3} \). \label{fig:backward-step2}}
\end{figure}

\paragraph{Step 3 (\( j = 2 \)).}
In order to construct \(S^{c}_{\geq 2}\), we consider the sets in $S^{c^*}_{\geq 3}$ and their extensions with job 2. This gives \(S^{c}_{\geq 2} = \{ \emptyset, \{3\}, \{2\}, \{2,3\}\}\). The profit functions for the sets that contain job 2 are:
\begin{align*}
B(\{2\},p) &= -150 + 0.5 \cdot 500 \cdot p\\
&= -150 + 250 \cdot p \\
B(\{2,3\},p) &= -150 - 70 + (0.5 \cdot 500 + 0.5 \cdot 0.5 \cdot 350) \cdot p \\
&= -220 + 337.5 \cdot p 
\end{align*}
The lower bound for the unknown joint success probability before job 2 equals $\pi_1=0.75$. Analyzing the curves in Figure~\ref{fig:backward-step3}, we observe that set \(\{3\}\) is optimal for \(P < 0.923\) and set \(\{2,3\}\) for \(P > 0.923\). As a result, the non-dominated candidates for \(S^*_{\geq 2}\) are \(S^{c^*}_{\geq 2} = \{ \{3\} , \{2,3\} \}\).

 \begin{figure}[htbp]
 \centering
\includegraphics{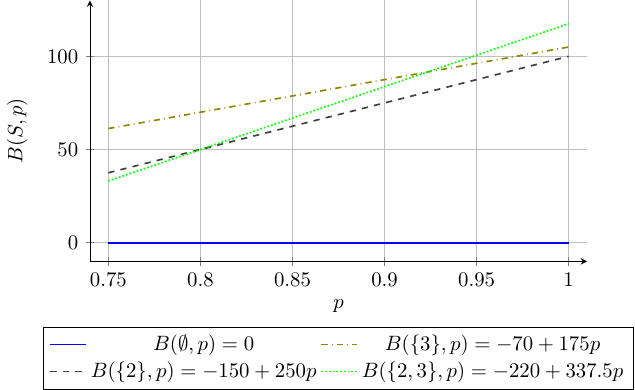}
\caption{Comparison of sets to get \( S^{c^*}_{\geq 2} \). \label{fig:backward-step3}}
\end{figure}

\paragraph{Step 4 (\( j = 1 \)).}
The set of candidate sets for \(S^*\) before comparing profit functions is given by $S^{c}_{\leq 1} =\{ \{3\}, \{2,3\}, \{1,3\}, \{1,2,3\} \}$. Since there are no jobs before job 1, the cumulative success probability is 1. Hence, we are simply left with evaluating the profit functions at \(p =1 \) for all candidate sets:
\begin{align*}
z(\{3\}) &= B(\{3\},p=1) = 105 \\
z(\{2,3\}) &= B(\{2,3\},p=1) = 117.5 \\
z(\{1,3\}) &= B(\{1,3\},p=1) = -75 - 70 + 0.75 \cdot 250 + 0.75 \cdot 0.5 \cdot 350 \\
&= 173.75 \\
z(\{1,2,3\}) &= B(\{1,2,3\},p=1) \\
&= -75 - 70 - 150 + 0.75 \cdot 250 + 0.75 \cdot 0.5 \cdot 500 + 0.75 \cdot 0.5 \cdot 0.5 \cdot 350 \\
&= 145.625.
\end{align*}
Among these sets,  \(\{1,3\}\) yields the highest profit and is therefore the optimal solution: \( S^* = \{1,3\} \). In total, the backward algorithm required evaluating eight out of the 16 possible subsets, namely:
\[
\emptyset, \{4\}, \{3\}, \{3,4\}, \{2\}, \{2,3\}, \{1,3\}, \text{ and }\{1,2,3\}.
\]

\subsection{Speed-ups}
Job \(j\) does not always need to be added to all subsets in \(\mathcal{S}^{c^*}_{\leq j - 1}\) or \(\mathcal{S}^{c^*}_{\geq j + 1}\). In the forward algorithm it will never be optimal to add job \(j\) to set \(S \in \mathcal{S}^{c^*}_{\leq j - 1}\) if this set is only optimal for values of \(r\) smaller than \(\pi_{j} \cdot r_{j}\), as adding this job already leads to a revenue that is too high. Similarly, in the backward algorithm we notice that it is never optimal to add job \(j\) to a subset \(S \in \mathcal{S}^{c^*}_{\geq j + 1}\) if this subset is only optimal for values of \(p\) greater than \( \pi_{j}\) as this leads to a probability that is too low.

\subsection{Keeping track of the upper envelope of the linear profit functions} \label{subsec:upenv}
During the execution of the forward and backward algorithms, we maintain a collection of non-dominated profit functions, denoted \(F^*\). These functions are all linear and can thus be written as \(y = a_i \cdot x + b_i\) where \(a_i = \prod_{i \in S} \pi_i\), \(b_i = R(S) - \sum_{i \in S} c_i\) and \(x = r\) in the forward algorithm and \(a_i = R(S)\), \(b_i = - \sum_{i \in S} c_i\) and \(x = p\) in the backward algorithm. Functions in \(F^*\) are ordered by non-decreasing slope such that \(a_1 \leq a_2 \leq \ldots \leq a_{|F^*|}\) and initially \(F^*\) contains only the profit function corresponding to the empty set. 

Whenever the range of possible values for either \(p\) or \(r\) decreases and \(|F^*| \geq 2\), we check whether we can remove some profit functions. If the upper bound \(\bar{x}\) on \(x\) goes down, we check if the intersection of profit function \(|F^*|\) and \(|F^*|-1\) has an x-value \(> \bar{x}\). If this is the case, we remove the last profit function in the set and repeat the process until this is no longer the case or only one function remains. If the lower bound \(\ubar{x}\) on \(x\) goes up, we check if the intersection of the first two lines has an x-value \(< \ubar{x}\). If this is the case, we remove the first profit function in the set and check again until this is no longer the case or only one profit function remains.

After updating the upper or lower bound on $x$, new profit functions are considered for inclusion into $F^*$. These profit functions are considered one by one as described in Algorithm~\ref{alg:profit}. This algorithm is partially based on the algorithm introduced by 
\cite{preparata1979}. However, we don't use search trees but vectors to store the profit functions, mainly for convenience. The algorithm works as follows. First, if the new function \(y = a^* \cdot x + b^*\) has the same slope as an existing function but a lower intercept, it is immediately discarded as dominated. Then, we calculate \(x_r\), the intersection point with the first existing function having a higher slope and \(x_l\) the intersection point with the last existing function having a lower slope. If \(x_r \leq x_l\) or the profit function is only non-dominated for non-feasible values of \(x\), the new function is discarded. If the profit function is not dominated for some feasible value of \(x\), we check if it dominates some other functions. Dominated functions are removed and the new function is added.

\begin{algorithm}[htbp]
    \caption{Adding \(y = a^* \cdot x + b^*\) to \(F^*\)} \label{alg:profit}
    \begin{algorithmic}
        \If{$a^* = a_i$ for some $i \in \{1,\ldots,k\}$}
            \If{$b_i \geq b^*$}
                \State \Return $F^*$ \Comment{New profit function is dominated}
            \EndIf
        \EndIf
	\State Find \(i_r\) and \(i_l\) and calculate \(x_r\) and \(x_l\) \Comment{See Algorithm \ref{alg:rl}}
        \If{\(x_r \leq x_l\) or \(x_r \leq \ubar{x}\) or \(x_l \geq \bar{x}\)}
            \State \Return $F^*$ \Comment{New profit function is dominated}
        \EndIf
	\State Update \(i_r\) and \(i_l\) \Comment{See Algorithm \ref{alg:dom}}
        \State Delete the profit functions \(\{i_l + 1,\ldots,i_r - 1\}\) from $F^*$
        \If{$i_l = 1$ and $x_l \leq \ubar{x}$}
            \State Delete profit function \(i_l\) from \(F^*\)
        \EndIf
        \If{$i_r = |S|$ and $x_r \geq \bar{x}$}
            \State Delete profit function \(i_r\) from \(F^*\)
        \EndIf
        \State Add profit function with intercept \(b^*\) and slope \(a^*\) to $F^*$
        \State \Return $F^*$
    \end{algorithmic}
\end{algorithm}

\begin{algorithm}[htbp]
\caption{Find \(i_r\) and \(i_l\) and calculate \(x_r\) and \(x_l\) for \(y = a^* \cdot x + b^*\)} \label{alg:rl}
\begin{algorithmic}
\If{\(a^* > a_{|F^*|}\)}
            \State $i_r \gets |F^*| + 1$ \Comment{There is no profit function with a higher slope}
            \State $x_r \gets +\infty$
        \Else
            \State $i_r \gets \min \{i \in \{1,\ldots,|F^*|\} : a_i > a^*\}$
            \State $x_r \gets \frac{b^* - b_{i_r}}{a_{i_r} - a^*}$
        \EndIf
        \If{$a^* < a_{1}$}
            \State $i_l \gets 0$ \Comment{There is no profit function with a lower slope}
            \State $x_l \gets -\infty$
        \Else
            \State $i_l \gets \max \{i \in \{1,\ldots,|F^*|\} : a_i < a^*\}$
            \State $x_l \gets \frac{b_{i_l} - b^*}{a^* - a_{i_l}}$
        \EndIf
\end{algorithmic}
\end{algorithm}

\begin{algorithm}[htbp]
\caption{See if \(y = a^* \cdot x + b^*\) dominates existing functions in \(F^*\) by updating \(i_r\) and \(i_l\)}  \label{alg:dom}
\begin{algorithmic}
\While{$i_l > 1$}
            \State $x'_l \gets \frac{b_{i_l - 1} - b^*}{a^* - a_{i_l - 1}}$
            \If{$x'_l \geq x_l$} \Comment{Intersection with function with lower slope has higher x-value}
                \State $x_l \gets x'_l$ and $i_l \gets i_l -1$
            \Else 
                \State Break
            \EndIf
        \EndWhile
        \While{$i_r < |F^*|$}
            \State $x'_r \gets \frac{b^* - b_{i_r + 1}}{a_{i_r + 1} - a^*}$
            \If{$x'_r \leq x_r$} \Comment{Intersection with function with higher slope has lower x-value}
                \State $x_r \gets x'_r$ and $i_r \gets i_r + 1$
            \Else
                \State Break
            \EndIf
        \EndWhile
\end{algorithmic}
\end{algorithm}

\newpage
\section{Computational experiments}\label{sec:comp}
In this section, we present the results of a computational study evaluating the performance of the proposed solution approaches. Two sets of experiments were conducted. The first uses instances generated according to a standard pattern for scheduling problems, adapted to the specific features of UJSSP\@. A second set of harder instances was generated from the PPP, following the reduction in Theorem~\ref{thm:comp}. This allows us to test our approach on a strongly NP-hard problem. The linear formulation described in Section \ref{sec:milp} was implemented using the solver Gurobi (version 12.0.2), and all code was written in C++. The source code, along with the data and results, is publicly available on GitHub \citep{perneel2025}. All computational experiments were performed on a single IceLake node of the wICE cluster (KU Leuven/UHasselt, VSC), equipped with 2 Intel Xeon Platinum 8360Y CPUs (36 cores each, 2.4 GHz), 256 GB of RAM (3.4 GB per core), and 960 GB local SSD storage.

\subsection{First dataset: uniform instances}
Hereafter, for a given instance of UJSSP, we let \(p\) denote the product of all job success probabilities, also called the \textit{joint probability}. For the first dataset, we generate instances of UJSSP as follows. For each job \(j \in J\), an integer revenue \(r_j\) is drawn uniformly from \([50,500]\). Success probabilities are generated under four schemes: 
(i) each job receives a random probability from \([0.01,0.99[\); 
(ii) a joint probability \(p \sim U[0.01,0.10[\) is distributed over the jobs using weights \(w_j \in [1,1000]\) drawn uniformly, with \(\pi_j = e^{\ln(p)\,w_j / W}\), where \(W = \sum_{j \in J} w_j\); 
(iii) as in (ii), but with \(p \sim U[0.10,0.40[\); and 
(iv) as in (ii), but with \(p \sim U[0.40,0.90[\). 
Finally, each job cost \(c_j\) is drawn uniformly from \([\lceil p\,r_j \rceil, \lfloor \pi_j\,r_j \rfloor]\), excluding values that make the inclusion of job \(j\) trivially optimal or suboptimal.We perform experiments for different values of the number of jobs \(|J|\). We consider small to medium-sized instances, where \(|J|\) ranges from 5 to 150 in increments of 5 (30 values), and large instances, where \(|J|\) ranges from 500 to 10{,}000 in increments of 500 (20 values). For each value of~\(|J|\), 40 instances are generated (10 for each probability generation scheme (i)--(iv)), which yields a total of \( (30 + 20) \times 40 = 2000 \) instances.

\subsubsection{MILP}
We assess the computational performance of the MILP formulation presented in Section~\ref{sec:milp} using the Gurobi optimizer, with a time limit of 20 minutes per instance. The results are summarized in Figure~\ref{fig:MILP}. All instances with up to 40 jobs are solved to optimality within the prescribed time limit, regardless of the probability generation scheme. Under scheme~(i), instances with up to 80 jobs are solved to optimality, and computation times start growing significantly around \(n = 100\). For the remaining schemes, optimization becomes considerably more challenging: computation times start increasing sharply when $n$ exceeds 50, and no instances with more than 75 jobs can be solved within the 20-minute time limit.

\begin{figure}[htbp]
\centering
\begin{subfigure}{0.48\textwidth}
\centering
\includegraphics[width=\textwidth]{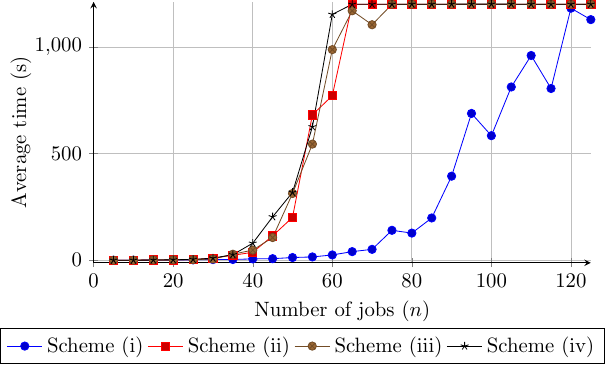}
\caption{Computation time}
\end{subfigure}
\hfill 
\begin{subfigure}{0.48\textwidth}
\centering
\includegraphics[width=\textwidth]{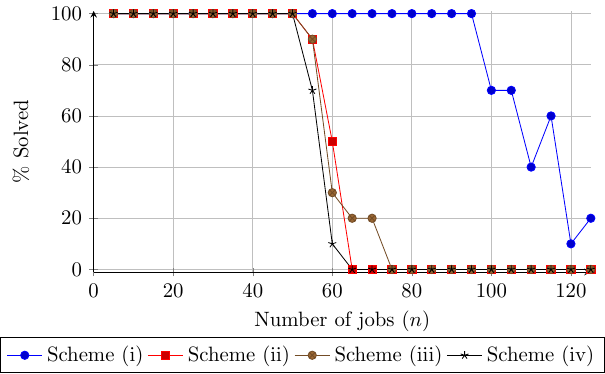}
\caption{Percentage solved}
\end{subfigure}
\caption{MILP: average runtime and percentage solved to optimality vs.\ number of jobs \label{fig:MILP}} 
\end{figure}

To further assess the quality of the MILP formulation, we analyze the \textit{LP\_gap} for each instance, computed as: 
\[
\text{LP\_gap} = \frac{\text{LP bound} - \text{Best found}}{\text{Best found}},
\]
where \textit{Best found} is the value of the best feasible solution obtained by the solver within the 20-minute limit (equal to the optimal solution when optimality is proven), and \textit{LP bound} is the objective value of the LP relaxation at the root node. We also report the \textit{final\_MIP\_gap} after the time limit:
\[
\text{final\_MIP\_gap} = \frac{\text{Best upper bound} - \text{Best found}}{\text{Best found}},
\]
where \textit{Best upper bound} is the best bound found by the solver after 20 minutes. Both metrics are shown in Figure~\ref{fig:MILP_Gap}. Interestingly, instances generated under scheme~(i) display notably larger LP gaps at the root node compared to the other schemes. Nonetheless, these instances tend to be easier to solve to optimality within the time limit. Overall, the results indicate that although the initial LP gaps can be substantial, the solver is typically able to substantially reduce the gap within the 20-minute time limit.

\begin{figure}[htbp]
\centering
\begin{subfigure}{0.48\textwidth}
\centering
\includegraphics[width=\textwidth]{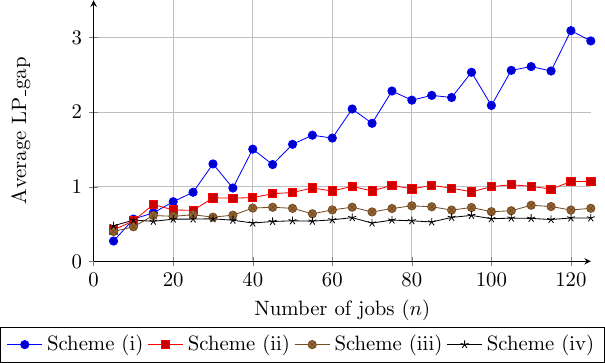}
\caption{Average LP\_gap}
\end{subfigure}
\hfill
\begin{subfigure}{0.48\textwidth}
\centering
\includegraphics[width=\textwidth]{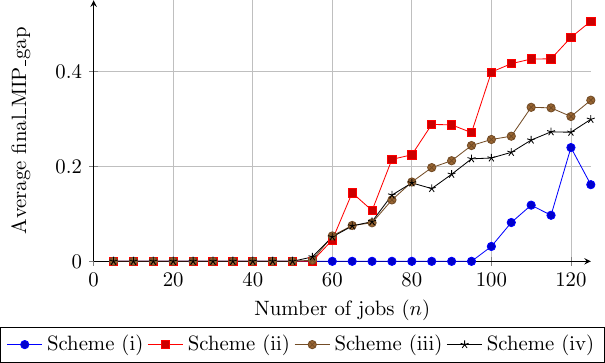}
\caption{Average final\_MIP\_gap}
\end{subfigure}
\caption{MILP: average LP\_gap and average final\_MIP\_gap vs.\ number of jobs \label{fig:MILP_Gap}}
\end{figure}

\subsubsection{Dynamic programming}
Next, we assess the performance of the DP algorithm introduced in Section~\ref{sec:dp}. Figure~\ref{fig:DP} summarizes the computation times for different instance sizes and probability generation schemes. All instances with up to 2{,}500 jobs are solved within 10 minutes, while the instances with 3{,}000 jobs often take more than one hour of computation. Consistent with the MILP results, instances where probabilities are generated through scheme~(i) are noticeably easier to solve. For the DP algorithm, it is more straightforward to explain this than for the MILP formulation. The DP algorithm runs in $O(n \cdot \sum_{j\in J} c_j )$ time. Each cost value \(c_j\) is drawn from the interval \([\lceil p\cdot r_j \rceil, \lfloor \pi_j \cdot r_j \rfloor]\). When probabilities are generated according to schemes~(ii)--(iv), value \(\pi_j r_j\) tends to approach \(r_j\) as \(n\) increases, while this is not the case under scheme~(i). In contrast, under scheme~(i), the expected value of \(p \cdot r_j\) will decrease as \(n\) grows, while this is not the case under schemes~(ii)--(iv). The average cost of a job will therefore be lower under scheme~(i) than under schemes~(ii)--(iv), especially as \(n\) grows. As the complexity of the DP algorithm grows linearly with the total cost of the jobs, instances generated under scheme~(i) are solved in shorter computation times.

\begin{figure}[htbp]
\centering
\includegraphics[width=0.48\textwidth]{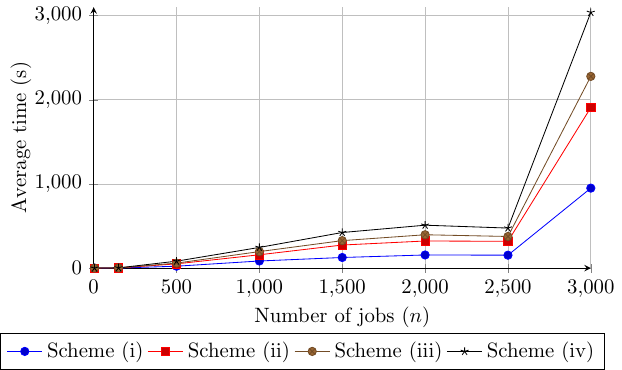}
\caption{DP: average runtime vs.\ number of jobs \label{fig:DP}}
\end{figure}

\subsubsection{Stepwise exact algorithms}
We now evaluate the performance of the forward and backward stepwise exact algorithms on the benchmark instances. This experiment includes instances with up to 10{,}000 jobs, which is the largest size generated in our study. The corresponding results are summarized in Figure~\ref{fig:stepw}. When job success probabilities are drawn according to scheme~(i), all instances are solved in less than one second of CPU time, even for \(n = 10{,}000\). For schemes~(ii)--(iv), instances with up to 1{,}500 jobs are solved within one second, after which computation times increase gradually, reaching approximately one minute for the forward algorithm and two minutes for the backward algorithm at \(n = 10{,}000\). Although the forward stepwise algorithm is generally faster on average, performance differences between the two methods remain instance-dependent.

\begin{figure}[htbp]
\centering
\begin{subfigure}{0.48\textwidth}
\centering
\includegraphics[width=\textwidth]{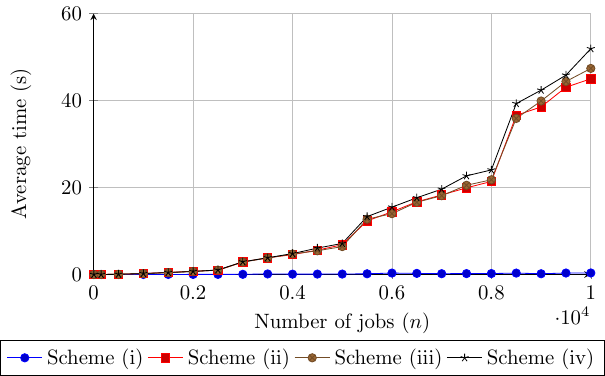}
\caption{Forward method}
\end{subfigure}
\hfill
\begin{subfigure}{0.48\textwidth}
\centering
\includegraphics[width=\textwidth]{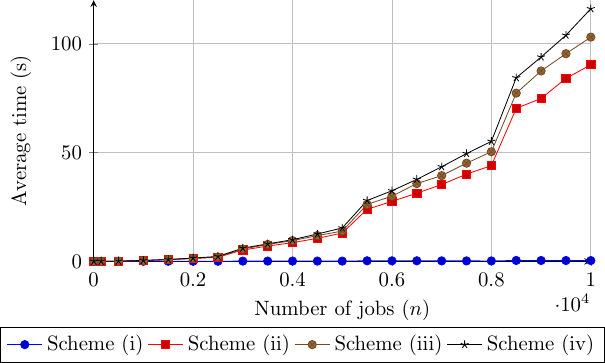}
\caption{Backward method}
\end{subfigure}
\caption{Forward and backward stepwise exact methods: average runtime vs.\ number of jobs \label{fig:stepw}}
\end{figure}

\subsection{Second dataset: Product Partition}
Although the PPP is a well-known combinatorial optimization problem and a standard source of hardness proofs, known to be strongly NP-hard \citep{ng2010}, it has not been studied as a computational benchmark in its own right. A straightforward solution approach applies logarithmic transformation to the input integers and scales the resulting values to form an equivalent instance of the Partition problem. However, achieving sufficient precision to distinguish between products can require integer magnitudes of the order of the square root of the total product. A more recent reduction proposed by \cite{costandin2024} produces an instance of the Subset Sum Problem having polynomial size in the PPP input, but requires integer factorization of the original numbers, which can itself be computationally expensive.

In this work, we present the first computational results for PPP instances, to the best of our knowledge. Rather than converting PPP into SSP, we generate UJSSP instances through the reduction introduced in Theorem~\ref{thm:comp} and solve them using the stepwise exact algorithms. Each optimal solution of the corresponding UJSSP instance yields a partition of the original integers whose subset products are as close as possible to the square root of the total product. Our main aim is to test the performance of our stepwise exact algorithms on more challenging instances. Note that the DP algorithm cannot be applied to these instances because the reduction in Theorem~\ref{thm:comp} produces non-integer costs.

When an instance of UJSSP is constructed from a PPP instance according to Theorem~\ref{thm:comp}, all \(Z\)-ratios are identical. Consequently, the order in which integers (or equivalently, jobs) are considered in the stepwise algorithms becomes irrelevant, and the forward and backward stepwise methods become equivalent. Consider a PPP instance defined by the set of integers \(a_1, \ldots, a_n\), arranged according to some permutation \(\sigma = (\sigma(1), \ldots, \sigma(n))\), where \(\sigma(k)\) denotes the \(k\)-th integer in the sequence. For any subset \(S_{\leq j} \subseteq \{\sigma(1), \ldots, \sigma(j)\}\), we define the associated profit function—corresponding to both \(F(\cdot; r)\) and \(B(\cdot; p)\)—as
\[
-\log\!\left(\prod_{a_i \in S_{\leq j}} a_i\right)
- \frac{\sqrt{W}}{\prod_{a_i \in S_{\leq j}} a_i}\, x,
\]
where \(x\) is bounded by 
\[
\frac{1}{\prod_{a_i \in \{\sigma(j),\ldots,\sigma(n)\}} a_i} \leq x \leq 1.
\]
If, for a given subset \(S_{\leq j}\), no feasible value of \(x\) yields a higher profit than other subsets of \(\{\sigma(1), \ldots, \sigma(j)\}\), that subset and all its extensions can be safely discarded, as they cannot lead to an optimal solution. Since the ordering of jobs does not affect the outcome, we evaluate three ordering strategies: sorting jobs in ascending order, descending order, and in random order of the corresponding integers.

Two types of PPP instances were generated; these were subsequently converted into UJSSP instances according to Theorem~\ref{thm:comp} to generate a second dataset of instances of UJSSP\@. The first type (\emph{Type~I}) does not guarantee the existence of a partition with equal products; each integer \(a_i\) is independently drawn from the uniform distribution on \([2, 100]\). The second type (\emph{Type~II}) is constructed to ensure that a perfect partition exists, as follows. We first draw a random integer \(j \in [1, n-1]\) and generate \(a_1, \ldots, a_j\) uniformly from \([2,100]\); let \(P = \prod_{i=1}^{j} a_i\). We then seek \((n-j)\) additional integers \(a_{j+1}, \ldots, a_n\) such that their product also equals \(P\). To achieve this, we first compute the prime factorization of \(P = p_1^{n_1} p_2^{n_2} \cdots p_k^{n_k}\) and iteratively group its prime factors into \((n-j)\) integer products, each within \([2,100]\). If factorization limits (\(\sum_i n_i < n-j\)) or value bounds (\(\sqrt[n-j]{P} > 100\)) make such grouping impossible, the instance is discarded. During construction, we randomly assign subsets of the remaining prime factors to each \(a_{j+i}\), ensuring feasibility. If no valid assignment is found, the instance is discarded. Instance sizes range from \(n = 5\) to \(100\) in increments of five. For each size, we generate 20 instances (10 of Type~I and 10 of Type~II), yielding a total of 400 instances.

The computational results of the stepwise methods for the second dataset are presented in Figure~\ref{fig:PPP}. For instances guaranteed to have a feasible partition (Type~II), all cases with up to 25 integers are solved within 20 minutes, whereas none with more than 50 complete in that time. 
For Type~I (randomly generated instances), solvability drops even faster: only instances with up to 20 integers are solved, with computation times increasing sharply thereafter. Among the tested orderings, processing integers in descending order consistently performs best, while ascending order generally performs worst. Compared to the first dataset, substantially smaller instances can be solved to optimality.  This increased difficulty stems from both structural and numerical factors. First, PPP instances produce far fewer dominated subsets, forcing the algorithms to evaluate many more candidate combinations (see Table~\ref{tab:num_comb}). Although pruning still eliminates a large fraction of subsets, the proportion of excluded candidates is markedly lower than in the first dataset. Second, instances in the second dataset require evaluating products of up to \(n\) integers between 2 and 100, leading to very large intermediate values and necessitating the use of arbitrary-precision arithmetic libraries. Operations using these libraries are more memory-intensive and computationally expensive than standard integer operations, explaining the longer runtimes even for a comparable number of evaluated combinations.

\begin{figure}[htbp]
\centering
\begin{subfigure}{0.48\textwidth}
    \centering
    \includegraphics[width=\textwidth]{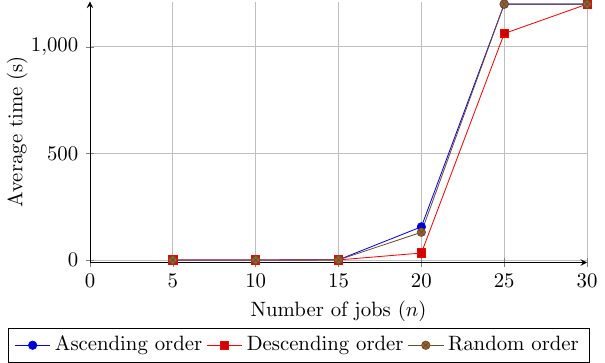}
    \caption{Type I: computation time}
\end{subfigure}
\hfill
\begin{subfigure}{0.48\textwidth}
    \centering
    \includegraphics[width=\textwidth]{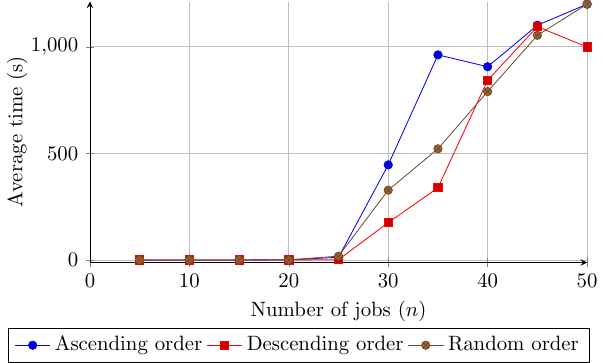}
    \caption{Type II: computation time}
\end{subfigure}
\vskip\baselineskip
\begin{subfigure}{0.48\textwidth}
    \centering
    \includegraphics[width=\textwidth]{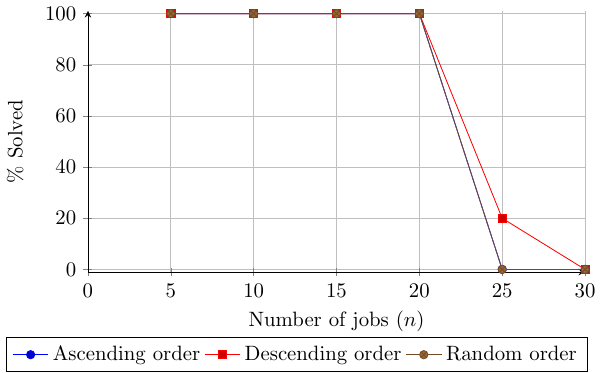}
    \caption{Type I: \% solved}
\end{subfigure}
\hfill
\begin{subfigure}{0.48\textwidth}
    \centering
    \includegraphics[width=\textwidth]{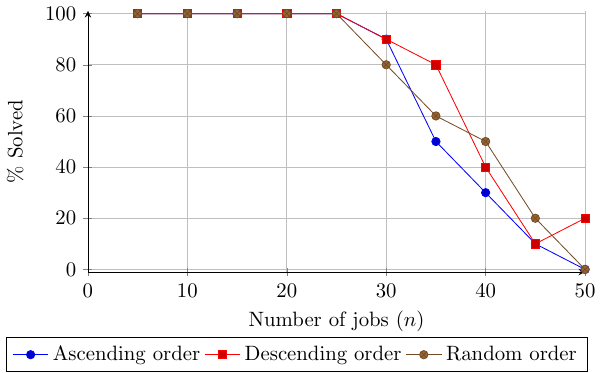}
    \caption{Type II: \% solved}
\end{subfigure}
\caption{Stepwise methods applied to dataset 2 (PPP): average runtime and \% solved to optimality vs.\ number of jobs \label{fig:PPP}}
\end{figure}

\begin{table}[htbp]
\centering
\begin{tabular}{@{}c@{\quad}c@{\quad}c@{\quad}c@{\quad}c@{\quad}c@{\quad}c@{}}
$n$ & PPP I & PPP II & UJSSP (i) & UJSSP (ii) & UJSSP (iii) & UJSSP (iv) \\ \hline
5  & 27.0 & 12.9 & 13.7  & 13.9  & 14.2  & 13.5 \\
10 & 587.3  & 181.6   & 31.3  & 33.4  & 33.2  & 35.4 \\
15 & 12,991.0 & 1,413.4 & 54.2  & 60.0  & 61.0  & 58.8 \\
20 & 210,402.0 & 14,488.1 & 78.3  & 82.3  & 89.2  & 92.2 \\
25 & -- & 75,220.9 & 101.9 & 119.9 & 121.2 & 135.0 \\
30 & -- & 469,513.1 & 137.8 & 185.0 & 160.4 & 172.8 \\ \hline
\end{tabular}
\caption{Average number of combinations considered using stepwise methods for different instance generation methods.\label{tab:num_comb}}
\end{table}

\section{Conclusions}\label{sec:conc}
In this paper, we have introduced the \emph{Unreliable Job Selection and Sequencing Problem} (UJSSP), a stochastic scheduling problem in which a subset of jobs must be selected and sequenced on a single machine subject to permanent failure. The objective is to maximize the expected net profit, accounting for job-specific rewards, costs, and probabilities of success. We have shown that UJSSP generalizes the \emph{$n$-out-of-$n$ Test Sequencing Problem}, is equivalent to the \emph{College Selection Problem}, and constitutes a special case of the \emph{Simultaneous Selection Problem}. These connections have enabled us to identify two polynomially solvable cases—when all job costs or all success probabilities are identical—where a greedy algorithm yields an optimal solution.

We have established the NP-hardness of the general UJSSP through a reduction from the strongly NP-hard \emph{Product Partition Problem} and proposed several exact solution approaches: a compact mixed-integer linear programming (MILP) formulation, a dynamic programming algorithm, and two novel stepwise exact algorithms. The dynamic programming approach runs in pseudopolynomial time, allowing us to rule out strong NP-hardness for the integer-cost variant of UJSSP. The stepwise algorithms iteratively build and prune candidate subsets by exploiting structural properties of the objective function.

Computational experiments demonstrate the effectiveness of the proposed methods. The MILP formulation solves instances with up to roughly 100 jobs, the dynamic programming algorithm handles up to 3{,}000 jobs, and the stepwise exact methods solve instances with up to 10{,}000 jobs within two minutes. Moreover, by applying the stepwise methods to instances derived from the Product Partition Problem, we provide, to the best of our knowledge, the first exact computational results reported for that problem.

Several directions for future research remain open. It is unclear whether the general UJSSP—with possibly fractional costs and rewards—is weakly or strongly NP-hard, and whether the greedy algorithm admits a constant-factor approximation in the general setting. The only similar result in this research area is a constant-factor approximation for the greedy algorithm applied to a somewhat similar problem, called Product Knapsack by \cite{pferschy2021}. The complexity of UJSSP with uniform rewards also remains to be established. Also, while the stepwise algorithms exhibit strong performance for standard UJSSP instances, their effectiveness on instances derived from the Product Partition Problem is more limited. More generally, an exploration of how stepwise methods that exploit the linear dependence of the objective value on a single unknown parameter for certain subproblems, might be extended to other combinatorial problems, or identifying the conditions under which they perform best, represents a promising direction for future research. Another relevant direction for future research is to consider a more general model in which selected jobs have to be allocated to multiple machines. In this case, even when $S$ is given, the resulting scheduling problem is NP-hard \citep{agnetis2009}, but practically efficient algorithms may be worth pursuing. 

\section*{Acknowledgement}
We would like to express our gratitude to dr.\ Ben Hermans for the idea of the dynamic programming algorithm.

\bibliographystyle{plainnat}
\bibliography{bib_file}

\end{document}